\documentclass[11pt]{article}
\usepackage[numbers]{natbib}
\usepackage{amssymb, latexsym, mathtools, amsthm}
\begingroup
    \makeatletter
    \@for\theoremstyle:=definition,remark,plain\do{%
        \expandafter\g@addto@macro\csname th@\theoremstyle\endcsname{%
            \addtolength\thm@preskip\parskip
            }%
        }
\endgroup
\usepackage{enumerate}
\usepackage{pgf,tikz}
\usepackage{pdfsync}
\usepackage{txfonts}
\usepackage[T1]{fontenc}
\usepackage{booktabs}
 \usepackage{graphicx}
\definecolor{dnrbl}{rgb}{0,0,0.3}
\definecolor{dnrgr}{rgb}{0,0.3,0}
\definecolor{dnrre}{rgb}{0.5,0,0}
\usepackage[colorlinks=true, citecolor=dnrgr, linkcolor=dnrre, urlcolor=dnrre] {hyperref}
\usepackage{xcolor}
\usepackage{parskip}
\usepackage{tabularx, colortbl}

\usepackage{geometry}
 \usepackage[scriptsize, up]{caption}
\usepackage{pgf,tikz}
\usetikzlibrary{decorations, decorations.pathmorphing}
\usetikzlibrary {shapes}

\theoremstyle{plain}
\newtheorem{thm}{Theorem}[section]

\newtheorem{lem}[thm]{Lemma}

\newtheorem{defi}[thm]{Definition}
\numberwithin{equation}{subsection}
\makeatletter

\let\c@table\c@figure
\makeatother

\newcommand{\Nat}{\mathbb{N}}

\newcommand{\restr}{\upharpoonright}  
\newcommand{\un}{\uparrow} 
\newcommand{\de}{\downarrow} 

\DeclarePairedDelimiter{\tuple}{\langle}{\rangle}

\newcommand{\bigo}[1]{\mathop{\bf O}\/\left({#1}\right)}

\newcommand{\KS}{Ku{\v{c}}era and Slaman\ }
\newcommand{\CHKW}{Calude, Hertling, Khoussainov and Wang\ }

\newcommand{\ml}{Martin-L\"{o}f }

\newcommand{\ie}{i.e.\ }
\newcommand{\ce}{c.e.\ }

\newcommand{\lce}{left-c.e.\ }

\newcommand{\pf}{prefix-free }

\renewenvironment{abstract}
 { \normalsize
  \list{}{
    \setlength{\leftmargin}{.0cm}%
    \setlength{\rightmargin}{\leftmargin}%
    }%
  \item {\bf \abstractname.} \relax}
 {\endlist}

 \makeatletter
\newtheorem*{rep@theorem}{\rep@title}
\newcommand{\newreptheorem}[2]{%
\newenvironment{rep#1}[1]{%
 \def\rep@title{#2 \ref{##1}}%
 \begin{rep@theorem}}%
 {\end{rep@theorem}}}
\makeatother
 \newreptheorem{thm}{Theorem}
\newreptheorem{lem}{Lemma}
\usepackage{titlesec}
 \titleformat{\paragraph}
{\normalfont\normalsize\em\bfseries}{\theparagraph}{1em}{$\blacktriangleright$\hspace{0.2cm}}
\titlespacing*{\paragraph}
{0pt}{2.25ex plus 1ex minus .2ex}{0.5ex plus .2ex}

\title{Kobayashi compressibility
\thanks{Barmpalias was supported by the 
1000 Talents Program for Young Scholars from the Chinese Government, grant no.\ D1101130.
Additional support was received by
the Chinese Academy of Sciences (CAS) and the Institute of Software of the CAS. 
Downey was supported by the
Marsden Fund of New Zealand. The authors wish to
thank the anonymous referees for various suggestions and corrections.}}

\author{George Barmpalias  \and Rodney G.~Downey}
\date{\today}
\begin{document}
\maketitle
\begin{abstract}
Kobayashi \cite{Kobayashi_rep}
introduced a uniform notion of compressibility of infinite binary sequences $X$ in terms of
relative Turing computations with sub-identity use of the oracle.
Given $f:\Nat\to\Nat$ we say that $X$ is $f$-compressible if there exists $Y$ such that
for each $n$ we compute $X\restr_n$ using at most the first $f(n)$ bits of the oracle $Y$.
Kobayashi compressibility 
has remained a relatively obscure notion, with the exception of some work on resource bounded
Kolmogorov complexity. The main goal of this note is to show that it
is relevant to a number of topics in current research on algorithmic randomness.

We prove that Kobayashi compressibility can be used in order to define \ml randomness,
a strong version of finite randomness and Kurtz randomness, strictly in terms of Turing reductions.
Moreover these randomness notions naturally correspond  to Turing reducibility, weak truth-table reducibility
and truth-table reducibility respectively.
Finally we discuss Kobayashi's main result from \cite{Kobayashi_rep} regarding the compressibility of
computably enumerable sets, and provide additional related original results.
\end{abstract}
\vspace*{\fill}
\noindent{\bf George Barmpalias}\\[0.5em]
\noindent
State Key Lab of Computer Science, 
Institute of Software, Chinese Academy of Sciences, Beijing, China.
School of Mathematics and Statistics,
Victoria University of Wellington, New Zealand.\\[0.2em] 
\textit{E-mail:} \texttt{\textcolor{dnrgr}{barmpalias@gmail.com}}\\[0.2em]
\textit{Web:} \texttt{\textcolor{dnrre}{http://barmpalias.net}}\par
\addvspace{\medskipamount}\medskip\medskip
\noindent{\bf Rodney G.~Downey}\\[0.5em]  
\noindent School of Mathematics and Statistics,
Victoria University of Wellington, New Zealand.\\[0.2em]
\textit{E-mail:} \texttt{\textcolor{dnrgr}{rod.downey@vuw.ac.nz}}\\[0.2em]
\textit{Web:} \texttt{\textcolor{dnrre}{http://homepages.ecs.vuw.ac.nz/$\sim$downey}} 

\vfill \thispagestyle{empty}
\clearpage

\section{Introduction}
The compressibility of a finite binary program $\sigma$ is defined in terms of the shortest program
that can generate $\sigma$. This is the idea behind the theory of Kolmogorov complexity $C$ of strings.
For example, if $c\in \Nat$ then  $\sigma$ is $c$-incompressible if $C(\sigma)\geq |\sigma|-c$, and similar definitions
are used with respect to the prefix-free complexity $K$, where the underlying universal machine is prefix-free.
This notion of incompressibility has a well-known extension to infinite binary streams $X$, where we say
that $X$ is $c$-incompressible if $K(X\restr_n)\geq n-c$ for all $n$. Then the algorithmic randomness of $X$ is often
identified with the property that $X$ is $c$-incompressible for some $c$, and coincides with the notion of \ml
randomness\footnote{This is the first robust and most accepted definition of algorithmic randomness and was originally introduced by \ml \cite{MR0223179} based on effective statistical tests.}. These concepts are basic in Kolmogorov complexity,
and the reader is referred to the standard textbooks \cite{Li.Vitanyi:93,rodenisbook} for the relevant background.

\subsection{Kobayashi compressibility and incompressibility}\label{fs5htfcyZb}
The reader may observe that the above extension of the definition of compressibility from finite to infinite sequences
is nonuniform, in the sense that the compression of the various initial segments of $X$ could be done by different, possibly unrelated short programs. A uniform extension of compressibility from strings to infinite streams would require
the individual short programs to be part of a single stream. Kobayashi \cite{Kobayashi_rep} considered exactly that approach.
\begin{defi}[Kobayashi \cite{Kobayashi_rep}]\label{QWlEEHxGsN}
Given $f:\Nat\to\Nat$ we say that $X$ is $f$-compressible if there exists $Y$ 
which computes $X$ via an oracle Turing machine which queries, for each $n$, 
at most the first $f(n)$ digits of $Y$
for the computation of $X\restr_n$.\footnote{The 
reader who is familiar with monotone complexity from Levin in \cite{levinthesis,Levin:73}
(also discussed in \cite[Section 3.15]{rodenisbook}) will note that if $X$ is $f$-compressible for a computable
function $f$, then $f$ is an upper bound on the monotone complexity of $X$.}
\end{defi}
Note that, since every real is computable from itself with identity use, Definition \ref{QWlEEHxGsN} only makes
sense if $f(n)$ occasionally dips well below $n$. This feature contrasts 
a standard caveat  that is often assumed in computability theory for convenience, 
that the oracle-use in relative computations is strictly increasing. 
%
Kobayashi did not necessarily require that $f$ is computable in this definition, but
added effectivity requirements in the statements of his results. 
We formulate the corresponding notion of incompressibility of a real $X$ based on
Definition \ref{QWlEEHxGsN} as follows.
\begin{defi}[Kobayashi incompressibility]\label{u7zzXYCGcJ}
We say that a real $X$ is Kobayashi incompressible if
it is not $f$-compressible for any function $f$ such that $n-f(n)$ is unbounded.
\end{defi}

Note that every set $X$ is $(n-c)$-compressible for every constant $c$.
Indeed, given $c$ one can consider $Y$ such that $X=X\restr_c\ast Y$, and by hardwiring $X\restr_c$ into a Turing
machine we can compute $X$ from $Y$ with oracle-use $n-c$.

Kobayashi showed that the class of the incompressible streams $X$ of Definition \ref{u7zzXYCGcJ} 
has measure 1.
We will see in the following that, in fact, this definition is equivalent to \ml randomness.
Furthermore, if Turing computability in this definition is replaced with stronger reducibilities, 
then we get alternative definitions of Kurtz randomness\footnote{originally from 
Kurtz \cite{Kurtz:81} and further studied by Wang \cite{wangthesis}
and Downey Griffiths and Reid \cite{566308}.}  
and a strong version of computably bounded randomness\footnote{introduced and studied by
Brodhead, Downey and Ng \cite{BrodheadDN12}.} which we call granular randomness,
respectively. We state these results in Section \ref{a94G2m4L32}, deferring their proofs in latter sections.
It is interesting to note that 
these alternative definitions do not involve measure or prefix-free machines, so they are unique
in that they only use notions from classical computability theory.
It is curious that Kobayashi's simple and natural notion of compressibility 
has remained rather obscure, and does not even feature in the
encyclopedic books on Kolmogorov complexity and computability \cite{Li.Vitanyi:93,rodenisbook, Ottobook}.\footnote{Of the
 two citations to Kobayashi's work in \cite{Li.Vitanyi:93} one is about a somewhat known result
regarding the structure of one-tape nondeterministic Turing machine time hierarchy and the other is \cite{ipl/Kobayashi93}.
Incidentally, the results in the latter paper were independently reproved  by 
Becher, Figueira, Grigorieff and Miller \cite{jsyml/BecherFGM06} (along with other original results).}

\subsection{Oracle-use in computations}\label{B1WlpjGBcu}
Note that if $f$ is non-computable, then 
the condition in Definition \ref{QWlEEHxGsN} does not necessarily mean that $X$ is computable
from $Y$ with oracle use $f$. The results we present often hide a
non-standard notion of oracle-use in computations, and for this reason we introduce some basic terminology.
We define {\em oracle-use} in a computation of $X$ from $Y$ through an oracle Turing machine in the
standard way, as the function $n\mapsto f(n)$ which 
indicates, for each $n$, the largest position in $Y$ which was
queried during the computation of $X(n)$. Note that this oracle use is computable in the oracle $Y$
(but in general non-computable), and
it is {\em adaptive}, in the sense that it depends on the oracle $Y$.
Another standard notion is the oracle-use of a truth-table reduction $X\leq_{tt} Y$.
In this case the {\em oracle-use of the  truth-table reduction} is the
function $n\mapsto g(n)$ which indicates, for each $n$, the 
the largest position in $Y$ which occurs in the truth-table corresponding to the computation of $X(n)$.
Note that the oracle-use of a truth-table reduction is computable and {\em oblivious} in the sense that
it does not depend on the oracle $Y$. Finally a weak-truth-table reduction 
$X\leq_{wtt} Y$ is exhibited by a Turing machine $M(n)$ and a computable function $h$ such that
the oracle-use of $M^Z(n)$ is bounded above by $h(n)$ for all oracles $Z$ and all numbers $n$
such that $M^Z(n)\de$.
In this case $h$ is called the oracle-use of the weak-truth-table reduction, and it is {\em oblivious}
and computable by definition.

We now introduce a non-standard definition.
Day \cite{day_process_tt} used the following notion in order to characterize various notions of
algorithmic randomness (see Section \ref{A8YTlyLStn}).
%
We say that $X$ is totally Turing reducible to $Y$ with oracle-use $f$ is there is a total Turing machine $M$
which computes $X$ with oracle $Y$ and oracle-use $f$.
Recall that  $X\leq_{tt} Y$ if and only if there is a total Turing machine $M$ 
(\ie such that $n\mapsto M^Z(n)$ is total for all $Z$) which computes $X$ with oracle $Y$.
The truth-table oracle-use is oblivious and computable 
while the  oracle-use of total reductions is adaptive and could be incomputable.
However the oracle-use of a total reduction has a computable upper bound, and it is computable in
the oracle $Y$. Day \cite{day_process_tt} provided 
characterizations of various notions of algorithmic randomness
based on the oracle-use in total reductions. We briefly discuss these contributions in Section 
\ref{A8YTlyLStn}, in the context of the present paper.


\subsection{Previous work on Kobayashi compressibility}\label{A8YTlyLStn}
An adaptation Kobayashi's compressibility notion for  resource-bounded computations was considered in
Balc{\'a}zar, Gavald{\`a}, and Hermo
\cite{balczar1994infinite,1996compressibility}, where it was shown that 
for logarithmic initial segment complexity (\ie when the plain or prefix-free Kolmogorov complexity
of the sequence is $\bigo{\log n}$),
the uniform and nonuniform approaches coincide,
both in the resource bounded and resource unbounded case. 
In particular, they showed that\footnote{The reader may recall that $\bigo{\log n}$ initial segment complexity
means the same (modulo additive constants) for plain or prefix-free Kolmogorov complexity, since the latter is at most
two times the plain complexity.}
\begin{equation*}
\parbox{13cm}{if $X$ has has initial segment complexity  $\bigo{\log n}$ then there
exists some $Y$ that computes $X$ with oracle-use $\bigo{\log n}$.}
\end{equation*}
Another observation from \cite{balczar1994infinite,1996compressibility} 
is that for constant oracle-use bounds and polynomial time,
Kobayashi's notion coincides with the polynomial computable functions and also with the corresponding resource-bounded
version in terms of initial segment complexity relative to the length of the initial segment, which was studied by
Loveland \cite{Loveland:69}.

Such uniform notions of compressibility for infinite streams based on Kobayashi's report were later used by
Balc{\'a}zar, Gavald{\`a} and Siegelmann in 
\cite{Balcazar:2006} in order to give a characterization of 
the computational power of recurrent neural networks in terms of the
Kolmogorov complexity of their weights. A hierarchy theorem regarding resource-bounded compressibility
notions based on functions $f$ of increasing growth rates is also included in 
\cite{Balcazar:2006}.

The notion of relative computations with occasionally 
sub-linear oracle-use that is behind Kobayashi's notion of compressibility
also initiated a study of strong reducibilities in a series of three never-before-cited papers by Habart 
\cite{Habart90, Habart91, Habart92}. From these papers, \cite{Habart90} is directly related to
algorithmic randomness notions.
Given a computable $f:\Nat\to\Nat$ let $\mathcal{S}(f)$ be the class of reals $X$ which are $f$-compressible according to
Definition  \ref{QWlEEHxGsN}. Given a subclass $\mathcal{C}$ of the class $\omega^{\omega}$ of functions
from the natural numbers to the natural numbers, Habart defined $X$ to be 
{\em $\mathcal{C}$-incompressible} (or incompressible with respect to $\mathcal{C}$) if 
\begin{equation}\label{g4LjTJRPt}
\forall f\in\mathcal{C},\ [X\in \mathcal{S}(f)\Rightarrow \mathcal{S}(f)=2^{\omega}]
\end{equation}
where $2^{\omega}$ is the class of all infinite binary sequences.
Similarly, $X$ is {\em $\mathcal{C}$-compressible} if it is not $\mathcal{C}$-incompressible, \ie if
there exists $f\in\mathcal{C}$ such that $X\in \mathcal{S}(f)$ and there are $f$-incompressible reals.
Habart then shows that 
\[
\parbox{9cm}{\ml randomness is equivalent to incompressibility with respect to all non-decreasing functions
in $\omega^{\omega}$.}
\]
Moreover Habart shows that the following conditions are equivalent for any real $X$:
\begin{enumerate}[\hspace{0.5cm}(a)]
\item $X$ is compressible with respect to the computable functions in $\omega^{\omega}$;
\item $X$ is compressible with respect to the nondecreasing computable functions in $\omega^{\omega}$;
\item there exists a computable function $g$ such that $K(X\restr_{g(n)})\leq g(n)-n$ for all $n$.
\end{enumerate}
The follow-up papers \cite{Habart91, Habart92} explore variations of the above notion of compressibility and, amongst other
results, give a characterization of bi-immunity in terms of compressibility.

Day \cite{day_process_tt}, unaware of of the above developments, gave characterizations of
computable randomness, Schnorr randomness and Kurtz randomness in Kobayashi's spirit of 
compressibility, based on the concept of oracle-use in total reductions as we discussed
in Section \ref{B1WlpjGBcu},
and effective martingales.
His results are in the same spirit as our Section \ref{a94G2m4L32}, but different due to the fact that
he considers total reductions. For example, he shows that  a real $X$ is computably random if
for every $Y$ which computes $X$ through a total reduction, the oracle use on $n$ is bounded above 
by $n-c$ for some constant $c$ and all $n$. His characterizations of Kurtz randomness and
Schnorr randomness are slightly different but in the same spirit.
Franklin, Greenberg, Stephan and Wu \cite{tinyuse}
showed that a reals $X$ is  $f$-compressible (recall Definition \ref{QWlEEHxGsN}) 
for all computable non-decreasing unbounded
functions $f$ if and only if for all such functions $f$ we have $K(X\restr_{f(n)})\leq n$ for almost all $n$.
Moreover they gave characterizations of this class of reals in terms of notions from classical computability
theory, such as \ce traceability.\footnote{The authors of \cite{tinyuse} use their own terminology, although
many of their results concern Kobayashi's notion of compressibility.}

The main result from \cite{Kobayashi_rep} is the following characterization of the compressibility of
computably enumerable (c.e.\ from now on) sets in terms of asymptotic bounds on the oracle use in relative computations.
\begin{thm}[Kobayashi \cite{Kobayashi_rep}]\label{TM6wyBIcXu}
If $f$ is a non-decreasing computable function then the following are equivalent:
\begin{enumerate}[\hspace{0.5cm}(a)]
\item $\sum_n 2^{-f(n)}$ is finite;
\item every \ce set $X$ is computable by some set $Y$ with oracle-use bounded above by $f$.
\end{enumerate}
\end{thm}
This is an elegant characterization of the compressibility of \ce sets.
Other results on this topic, such as
the work in \cite{BarzdinsCe, DBLP:journals/siamcomp/Kummer96, apal/BarmpaliasL13,ipl/BarmpaliasHLM13},
tend to focus on the Kolmogorov complexity of the initial segments of \ce sets.
In fact, Theorem \ref{TM6wyBIcXu} remains true if we add a third clause saying
`every \ce set $X$ is computable by $\Omega$ with oracle-use bounded above by $f$' where 
$\Omega$ is Chaitin's halting probability.  We give a short proof of this slightly enhanced 
result in Section \ref{8J69q48pEH}.

In the last decade, a number of vaguely related results in the spirt of Theorem \ref{TM6wyBIcXu} have emerged in the
literature. Solovay \cite{Solovay:75} (also see \cite[Section 3.13]{rodenisbook}) and 
Tadaki \cite{Tadaki:2009_72}
considered the problem of how many bits of $\Omega$ are needed in order to compute
the domain of the universal \pf machine $U$ up to the strings of length $n$. Tadaki showed that, 
if we restrict the question to computable use-functions, then 
the answer is given by the computable functions of the type $n-f(n)$ 
such that $\sum_n 2^{-f(n)}$ is finite. Similar asymptotic conditions on the oracle use in relative computations between
\ce reals were obtained in \cite{asybound}.

A typical case in Theorem \ref{TM6wyBIcXu} is for logarithmic bounds. We get that if $\epsilon>1$ 
then every \ce set is $(\epsilon\cdot\log n)$-compressible while  there are \ce sets that are not 
$(\log n)$-compressible. A result with the same flavor was recently presented in \cite{omefom}, where
it was shown that halting probabilities of universal \pf machines are computable from eachother with use
$n+\epsilon\cdot\log n$ for any $\epsilon>1$, but this statement is no longer true for $\epsilon=1$.

\subsection{Our results}\label{a94G2m4L32}
There are two types of results that we present, both of which are related to Kobayashi's report 
\cite{Kobayashi_rep}. First, in Section \ref{MJxfHSbpZm} we consider versions of Kobayashi incompressibility
and relate them to known randomness notions. Second, in Section \ref{vi1sgKVeb6} we
consider some results related to the initial segment complexity of \ce sets and Theorem \ref{TM6wyBIcXu}.

\subsubsection{Kobayashi incompressibility and reductions}\label{MJxfHSbpZm}
Our first result is  that Kobayashi incompressibility coincides with \ml randomness. This fact has also been noticed
by Laurent Bienvenu (personal communication---unpublished). Recall that a
set of strings $U$ can be viewed as the set of reals that have a prefix in $U$. In this way, the
measure $\mu(U)$ is the Lebesgue measure of the corresponding open set of reals.
Also recall that a \ml test is 
a uniformly \ce sequence of sets $(U_i)$ with $\mu(U_i)<2^{-i}$
for all $i$, and a
real $X$ is \ml random
if it is not in $\cap_i U_i$ for any \ml test $(U_i)$. 
\begin{thm}[Kobayashi incompressibility and Turing reductions]\label{6dkUL3ULut}
The following are equivalent:
\begin{enumerate}[\hspace{0.5cm}(a)]
\item  $X$ is \ml random;
\item for every $Y$ with $X\leq_T Y$ the $Y$-use in any such computation of $X\restr_n$
is bounded below by $n-c$ for some constant $c$ and all $n$.
\end{enumerate}
\end{thm}
Next, we consider the notion of  Kobayashi compressibility, as it was discussed in Section
\ref{fs5htfcyZb}, but with respect to strong reducibilities like truth-table and weak truth-table
reducibility. 
For example, 
given a computable $f:\Nat\to\Nat$ we can say that $X$ is truth-table $f$-compressible 
if there exists $Y$ which
truth-table computes $X$ with oracle-use bounded above by $f$.
Similar definitions apply with respect to weak truth-table reducibility.
It is now natural to ask whether these restricted notions of incompressibility characterize known
variants of algorithmic randomness. It turns out that this is indeed the case, and for the truth-table
variant the corresponding notion is Kurtz randomness, which was originally introduced in \cite{Kurtz:81}.
A \ml test $(U_i)$ is a {\em Kurtz test} if $U_i$ if finite for each $i$.\footnote{This formulation is due to 
Wang \cite{wangthesis}.}
A real
is {\em Kurtz random}
if it is not in $\cap_i U_i$ for any Kurtz test $(U_i)$.

\begin{thm}[Kobayashi incompressibility and truth-table reductions]\label{22L9TJlXkk}
The following are equivalent:
\begin{enumerate}[\hspace{0.5cm}(a)]
\item  $X$ is Kurtz random;
\item for every $Y$ such that $X\leq_{tt} Y$ the $Y$-use in any such computation of $X\restr_n$
is bounded below by $n-c$ for some constant $c$ and all $n$.
\end{enumerate}
\end{thm}
Downey, Griffiths and Reid \cite{566308} gave alternative characterizations of Kurtz randomness
in terms of tests and initial segment complexity. These characterizations can be used to give an alternative
proof of Theorem \ref{22L9TJlXkk}, as we note in the following.

If we consider Theorem \ref{6dkUL3ULut} with respect to weak truth-table reducibility, the corresponding
randomness notion
is what we call granular randomness in Theorem \ref{qTEYegz9as}, 
which expresses it as a restricted version of \ml
randomness in terms of tests, martingales, or \pf complexity.
Given a function $g$, we say that a 
\ml test $(V_i)$ is $g$-granular if for each $i$ each string in $V_i$ has length at most $g(i)$.
Moreover we say that a \ml test is computably granular if it is $g$-granular for some
computable increasing function $g$.

\begin{thm}[Granular randomness]\label{qTEYegz9as}
The following are equivalent for a real $X$:
\begin{enumerate}[\hspace{0.5cm}(a)]
\item $X\in \cap_i V_i$ for a computably granular \ml test $(V_i)$;
\item there exist a \ce (super)martingale $M$ and
a computable function $g$ such that
$\forall i\ M(X\restr_{g(i)})\geq i$;
\item there exists a computable function $g$
such that 
$\forall i\ K(X\restr_{g(i)})\leq g(i)-i$.\footnote{The referee has pointed out that this clause remains valid
if stated in terms of plain Kolmogorov complexity.}
\end{enumerate}
A real which does not have these properties is called granularly random.
\end{thm}
Note that, given a granular test 
$(V_i)$,
without loss of generality we may assume that 
for each $i$, all strings in $V_i$ have length $g(i)$. Moreover since $\mu(V_i)\leq 2^{-i}$ for each $i$,
we may also assume that $|V_i|\leq 2^{g(i)-i}$ for each $i$. 
This means that
granular randomness is
weaker\footnote{in fact, strictly weaker. By diagonalizing against all granular \ml tests it is a simple exercise to construct a
granularly random real which is not random with respect to any of the finite randomness notions of
\cite{BrodheadDN12}.} than
all versions of finite randomness that was introduced and studied by 
 Brodhead, Downey and Ng in \cite{BrodheadDN12}. On the other hand,
 granular randomness is stronger\footnote{in fact, strictly stronger. By diagonalizing against all Kurtz 
 tests it is a simple exercise to show that
there are Kurtz random reals which are not granularly random.} 
than Kurtz randomness. 

\begin{thm}[Kobayashi incompressibility and bounded oracle-use]\label{EFigRO1nBQ}
The following are equivalent:
\begin{enumerate}[\hspace{0.5cm}(a)]
\item  $X$ is granularly random;
\item for every $Y$ such that $X\leq_{\textrm{wtt}} Y$ the $Y$-use in any such computation of $X\restr_n$
is bounded below by $n-c$ for some constant $c$ and all $n$.
\end{enumerate}
\end{thm}
This last result is related to a result from 
Habart \cite{Habart90}, which asserts that
for each real $X$, property (c) of Theorem \ref{qTEYegz9as} is equivalent to $X$ being 
$\mathcal{R}$-compressible in the sense of definition \eqref{g4LjTJRPt}, where $\mathcal{R}$
here denotes the class of computable functions.

\subsubsection{Oracle-use and initial segment complexity of \ce sets and \ce reals}\label{vi1sgKVeb6}
The number $\Omega$
was introduced by Chaitin in \cite{MR0411829}
as the halting probability of a universal \pf machine, who also showed that it is \ml random.
Clearly $\Omega$ depends on the underlying universal \pf machine, but
the cumulative work of Solovay \cite{Solovay:75},
\CHKW \cite{Calude.Hertling.ea:01} and \KS \cite{Kucera.Slaman:01}
showed that these numbers are exactly the \lce \ml random numbers.
Finally we point out the following fact regarding \lce reals with highly complex initial segment.
As the universal halting probability, $\Omega$ computes all \lce reals including the halting problem itself. 
The following theorem qualifies the latter completeness property and relates to
the main result
of Kobayashi \cite{Kobayashi_rep}, Theorem \ref{TM6wyBIcXu}.

\begin{thm}\label{Etk37fTeOu}
Every \lce real $X$ can be computed from $\Omega$ with oracle-use
$g(n)=\min_{i\geq n} K(X\restr_n)$. 
\end{thm}
Recall that $K(\sigma)\leq |\sigma|+K(|\sigma|)$ for all strings $\sigma$,
and that if $f$ is a computable function such that $\sum_n 2^{-f(n)}$ is finite,
then $K(n)\leq f(n)+\bigo{1}$ for all $n$. Hence Theorem 
\ref{Etk37fTeOu} implies that any \lce real is computable from $\Omega$ with oracle-use $n+f(n)$,
where $f$ is any computable non-decreasing function such that $\sum_n 2^{-f(n)}$ is finite.
The latter fact was shown in \cite{asybound}, where it was also shown to be optimal in the sense that
if $\sum_n 2^{-f(n)}$ is infinite then there exists a \lce real which is not computable from $\Omega$ 
with oracle-use $n+f(n)$. 

It would be interesting to find in which sense 
the upper bound of Theorem \ref{Etk37fTeOu} for individual \lce reals $X$ is optimal.
This might not be straightforward as computable \lce reals are computable from any $\Omega$ with zero
oracle-use and \ml random \lce reals are versions of $\Omega$ and are computed from themselves with identity
oracle-use. With regard to the latter case, we do know from \cite{omefom} that given any version of $\Omega$
we can find a random \lce real which is not computable from $\Omega$ with use $n+\log n +\bigo{1}$.

Barzdins \cite{BarzdinsCe} showed that the plain complexity of 
(the characteristic sequence of) every \ce set is bounded above by
$2\log n +\bigo{1}$. Then
Kummer \cite{DBLP:journals/siamcomp/Kummer96} 
showed that this upper bound is optimal, in the sense that there exist \ce sets $A$ such that
for some constant $c$ we have $C(A\restr_n)> 2\log n -c$ for infinitely many $n$. Here 
we present a corresponding optimal monotone upper bound for the \pf complexity of \ce sets.

\begin{thm}\label{qVH5QUn2qI}
If $h$ is a computable non-decreasing function then the following are equivalent:
\begin{enumerate}[\hspace{0.5cm}(a)]
\item the \pf complexity of every \ce set is bounded above by $h(n)+\log n + \bigo{1}$;
\item $\sum_n 2^{-h(n)}$ is finite.
\end{enumerate}
\end{thm}

Recall from \cite{DJS2}
that a Turing degree $\mathbf{a}$ is array computable if there exists a function which is truth-table
reducible to the halting problem, which dominates all functions computable from $\mathbf{a}$.
Kummer \cite{DBLP:journals/siamcomp/Kummer96}  showed that array non-computable
\ce degrees contain \ce sets $A$ such that  $C(A\restr_n)> 2\log n -c$ for some constant $c$ and 
infinitely many $n$, while the plain complexity of the \ce sets of an array computable degree is bounded above
by $\log n + g(n)+\bigo{1}$ for any given non-decreasing unbounded computable function $g$.
This is sometimes known as the Kummer gap theorem (see \cite[Section 16.1]{rodenisbook}).

The \pf version of the Kummer gap theorem says that the upper bound 
$f(n)+\log n$
of Theorem \ref{qVH5QUn2qI} is tight for array non-computable \ce degrees 
(in the same sense that $2\log n$ was tight in the plain complexity case) and that $f(n)+g(n)+\bigo{1}$
is an upper bound for the \pf complexity of array computable \ce sets (where again, $g$ is any fixed 
computable unbounded non-decreasing function).
We state a representative case of this result, while the interested reader can verify that the proof that
we give in Section \ref{miPvnCjZSZ} applies to the general case.

\begin{thm}[Kummer gap for \pf complexity of \ce sets]\label{cJbypbgUJ}
If $\mathbf{a}$ is a \ce degree then
\begin{enumerate}[\hspace{0.5cm}(a)]
\item if $\mathbf{a}$ is array non-computable then it contains a \ce set whose \pf complexity is not bounded
above by $2\log n+\log\log n$;
\item if $\mathbf{a}$ is array computable then the \pf complexity of all of its \ce members is bounded above by
$\log n+2\log\log n+\bigo{1}$.
\end{enumerate}
\end{thm}

Note that the upper bound for the case of array non-computable degrees is slightly higher than Kummer's
$2\log n$, as it follows from Theorems \ref{cJbypbgUJ} and \ref{qVH5QUn2qI}.
The same is true for the case of array computable degrees, where
Kummer's bound was $\log n+g(n)$ for an arbitrary unbounded computable non-decreasing function $g$.


\section{Characterizations of randomness notions}
The arguments we give for
Theorems \ref{6dkUL3ULut}, \ref{22L9TJlXkk} and \ref{EFigRO1nBQ}.
are based on proving $\neg$(a) $\leftrightarrow$ $\neg$(b) in each case.
We see two different ways of
proving $\neg$(a) $\leftarrow$ $\neg$(b) in each of these cases.
One is to construct $Y$ and the required reduction of $X$ to $Y$ based on the
definition of the randomness notion at hand, in terms of tests.
The second way is to use the Kolmogorov complexity definition of the
corresponding randomness notion.
We demonstrate the first methodology in the proof of Theorem \ref{6dkUL3ULut} and
demonstrate the second methodology in the proofs of Theorems 
\ref{22L9TJlXkk} and \ref{EFigRO1nBQ}.

\subsection{Proof of Theorem \ref{6dkUL3ULut}}\label{vdtpKQ7fpR}
We need to show that the following are equivalent:
\begin{enumerate}[\hspace{0.5cm}(a)]
\item  $X$ is \ml random;
\item for every $Y$ such that $X\leq_T Y$ the $Y$-use in any such computation of $X\restr_n$
is bounded below by $n-c$ for some constant $c$ and all $n$.
\end{enumerate}
It suffices to prove $\neg$(a) $\leftrightarrow$ $\neg$(b).
First, assume that $X$ is not \ml random. Then for every constant $c$ there exists some $n$
such that $K(X\restr_n)\leq n-c$. Hence for every $c$ there exists some $n$ and a description of
$X\restr_n$ of length at most $n-c$. We can use this fact in order to define a stream $Y$ which 
computes $X$ with oracle-use some function $g$ such that 
$\limsup_n \left(n-g(n)\right)=\infty$.
We define $Y$ by induction, as the concatenation of a suitably crafted sequence of strings $(\sigma_i)$.
Let $\sigma_0$ be the empty sequence and let $n_0=c_0=0$. 
Inductively assume that for some $k>0$ all 
$\sigma_i, i<k$ have been defined
and belong to the domain of the universal \pf machine $U$. 
Let $c_k$ be the sum of all $|\sigma_i|$, $i\leq k$. 
Let $n_k>n_{k-1}$ be a number such that $K(X\restr_{n_k})\leq n_k-c_{k-1}-k$
and let $\sigma_k$ be the shortest description of $X\restr_{n_k}$. This concludes the definition of $(\sigma_i)$.

Let $Y=\sigma_0\ast\sigma_1\ast\cdots$ and note that, since all $\sigma_i$ are in the domain of $U$ and $U$
is a \pf machine, for each $k$ we can uniformly compute $X\restr_{n_k}$ from $Y\restr_{c_k}$.
We now show by induction that for each $k$ we also have $c_k+k\leq n_k$.
Clearly this holds for $k=0$. Inductively assume that $k>0$ and $c_{k-1}+k-1\leq n_{k-1}$. Then
$c_k=c_{k-1}+K(X\restr_{n_k})\leq c_{k-1}+ n_k-c_{k-1}-k=n_k-k$ which concludes the induction step.
Note that since $U$ is \pf the oracle $Y$ can compute the sequence $(c_i)$. Therefore for each $n$
the string $X\restr_n$ can be uniformly computed by $Y$ with use $c_k$, where $k$ is the least number such that $n\leq n_k$. So assuming  that $\neg$(a) we have proved that $\neg$(b).

For the other direction\footnote{The referee has pointed out that this direction also follows directly from
the characterization of \ml randomness in terms of monotone complexity from Levin \cite{MR0366096}
and Schnorr \cite{Schnorr:73}. See
\cite[Theorem 6.3.10]{rodenisbook}.} assume that $\neg$(b) and let $Y$ be an oracle such that  $X=\Phi(Y)$ for
a Turing functional $\Phi$ such for each $c$ there exists $n$ such 
that the $Y$-use $\phi(X\restr_n)$ for the computation of $X\restr_n$ is less than $n-c$.
We can view $\Phi$ as a \ce set of tuples $\tuple{\sigma,\tau}$ which indicate that any
stream with prefix $\sigma$ is $\Phi$-mapped to a string or stream with prefix $\tau$.
We enumerate a \ml test $(V_i)$ as follows. Given $k$, for each
$\tuple{\sigma,\tau}$ in $\Phi$ such that $|\tau|>k+|\sigma|$ we enumerate $\tau$ into $V_k$.
By our assumption we have $X\in\cap_i V_i$, so it remains to show that $\mu(V_k)\leq 2^{-k}$ 
for each $k$. For a contradiction, suppose that $\mu(V_k)> 2^{-k}$ for some $k$.
Then there exists a finite \pf set of strings $\{\tau_i\ |\ i<t\}\subseteq V_k$ such that
\[
\sum_{i<t} 2^{-|\tau_i|}>2^{-k}.
\]  
By the way we enumerate $V_k$ and since $\Phi$ is a Turing functional, 
for each $\tau_i, i<t$ there exists some $\sigma_i$ such that
$\tuple{\sigma_i,\tau_i}\in\Phi$, $|\sigma_i|< |\tau_i|-k$ and the set $\{\sigma_i\ |\ i<t\}$ is prefix-free.
So 
\[
\sum_{i<t} 2^{-|\tau_i|}< \sum_{i<t} 2^{-|\sigma_i|-k}\leq 2^{-k}\cdot \sum_{i<t} 2^{-|\sigma_i|}\leq 2^{-k}
\]  
which contradicts the previous inequality. We can conclude that 
$\mu(V_k)\leq 2^{-k}$ for each $k$, which completes the proof that $X$ is not \ml random. 

\subsection{Proof of Theorem \ref{22L9TJlXkk}}\label{DFdfWxKDG}
We need to show that the following are equivalent:
\begin{enumerate}[\hspace{0.5cm}(a)]
\item  $X$ is Kurtz random;
\item for every $Y$ such that $X\leq_{tt} Y$ the $Y$-use in any such computation of $X\restr_n$
is bounded below by $n-c$ for some constant $c$ and all $n$.
\end{enumerate}
We show $\neg$(a) $\leftrightarrow$ $\neg$(b). 
Note that for $\neg$(a) $\to$ $\neg$(b)
 one can use the characterization of Kurtz randomness in terms of a restricted version of \pf complexity
from Downey, Griffiths and Reid \cite{566308} (just as we used such a characterization of \ml randomness
in Section \ref{vdtpKQ7fpR}).
Instead, here we use the formulation of 
Kurtz randomness in terms of tests, originally by Wang \cite{wangthesis}.
 
\paragraph{Proof that $\neg$(a) implies $\neg$(b).}{Assuming that $X$ is not Kurtz random,
there exists a computable sequence $(D_i)$ of finite sets of strings  (also known as a strong array) such that
$\mu(D_i)\leq 2^{-i}$  and $X\in D_i$ for each $i$. Moreover, without loss of generality we may assume that
for each $i$, the strings in $D_i$ have the same length $d_i$. 
For each $i$, the number of strings in $D_i$ is at most $2^{d_i-i}$.
Without loss of generality we may assume that for each $i$ we have 
$d_i<d_{i+1}$, $|D_i|=2^{d_i-i}$,  $d_0>1$ and
every string in $D_{i+1}$ has a prefix in $D_i$.
It suffices to construct a total Turing functional $\Phi$ such that for each $X\in \cap_i D_i$ there
exists some $Y$ such that $X=\Phi^Y$. As a basis for this functional,
we define a partial computable tree $T$ as a partial computable function
from strings to strings,
and an increasing sequence $(q_s)$ of indices of the array $(D_i)$,
 such that
 \begin{enumerate}[\hspace{0.5cm}(a)]
\item for each $s$ and each $\rho\in D_{q_s}$ we have $T(\rho)\de$ and $|T(\rho)|=|\rho|-s$;
\item if $T(\rho)\de$ then $\rho\in D_{q_s}$ for some $s$.
\end{enumerate}
Let $\lambda$ be the empty string.
We define the sequence $(q_s)$ along with $T$ and 
an auxiliary increasing computable sequence $(p_s)$ by simultaneous
recursion, and use $p_s$ as the length
of $T(\rho)$ for each $\rho\in D_{q_s}$. In this construction we view the image
$T(\rho)$ of $\rho$ as a code for $\rho$.

At stage 0, let $q_0=p_0=0$ and let $T(\lambda)=\lambda$.
At stage $s+1$ assume inductively that
$q_j, p_j, j\leq s$ have been defined. Then let $q_{s+1}=p_s+s+1$
and $p_{s+1}= d_{q_{s+1}}-s-1$.
Next, for each string $\rho\in D_{q_{s}}$ map the extensions of $\rho$ in $D_{q_{s+1}}$ 
in lexicographical order onto
the extensions of $T(\rho)$ of length $p_{s+1}$. Formally, let
let $\eta_t, t<k$ be a lexicographical list of all the strings in $D_{q_{s+1}}$ which
extend $\rho$. Moreover let
$\theta_t,t<k$ be a lexicographical list of the first $k$
extensions of $T(\rho)$ of length $p_{s+1}$ and define $T(\eta_t)=\theta_t$ for each $t<k$.
This completes stage $s+1$ and the inductive definition of $T$.

It is straightforward to verify that
$T$ is well defined, \ie that at step $s+1$ 
the required assignment of strings is possible.
Moreover the tree $T$ is clearly a 
computable map from strings of length $d_{q_{s+1}}$ to strings
of length $p_{s+1}= d_{q_{s+1}}-s-1$. So the properties (a), (b) above hold. Since step $s+1$ always
maps distinct extensions of $\rho$ of a certain length to distinct  extensions of $T(\rho)$
of a certain length, the map $T$ is a tree in the sense that for $\eta,\theta$ such that $T(\eta)\de,T(\theta)\de$
we have
$\eta\subseteq\theta$ 
if and only if $T(\eta)\subseteq T(\theta)$.
Since $T$ is a tree and the mapping in step $s+1$ was done in lexicographical order,
we have that for each $s$ and each real  $X\in \cap_{j\leq q_{s+1}} D_{j}$ 
the code $T(X\restr_{d_{q_s}})$ is defined and
for all $s$ the string $X\restr_{d_{q_s}}$ is uniformly computable from $T(X\restr_{d_{q_s}})$
which has length $d_{q_{s}}-s$.
For each real $X\in \cap_i D_i$ let $T(X)$ be the limit of all $T(X\restr_{d_{q_s}})$ as $s$ goes to infinity. 
Consider
the intervals $I_s=[d_{q_s}, d_{q_{s+1}})$.
Moreover let $\Phi$ be the total Turing functional 
such that for each $X,s$ and $n\in I_s$ the oracle-use
for the first $n$ bits of $X$ is $d_{q_{s}}-s$ and
\[
\Phi^X(n)=
\begin{cases}
T(X)(n), &\textrm{if $X\in \cap_{j\leq q_{s+1}} D_{j}$;}\\
0,& \textrm{otherwise.}
\end{cases}
\]
Then we can conclude that each $X\in \cap_i D_i$ is truth-table reducible to $T(X)$
and for each $s$ the oracle-use for the computation  of the first $d_{q_{s}}$
bits of $X$ is $d_{q_{s}}-s$. This concludes the proof of $\neg$(b) given $\neg$(a).}

\paragraph{Proof that $\neg$(b) implies $\neg$(a).}{
Assume that there exists some $Y$ such that $X\leq_{tt} Y$ and for each $c$ 
there exists $n$ such that the oracle-use in this reduction for the computation of $n$-bit strings 
is bounded above by $n-c$. Let $\Phi$ denote a total Turing functional corresponding to the given 
truth-table reduction  $X\leq_{tt} Y$. 
Also, let $(n_i)$ be a computable increasing sequence of numbers such that for
each $i$ the oracle-use with respect to $\Phi$ for the computation of $n_i$-bit strings is bounded above
by $n_i-i$. Then we may assume that the oracle-use in computations with respect to $\Phi$
is a non-decreasing computable function $f$ such that $f(n_i)\leq n_i-i$ for each $i$, which is
the same for all oracles.
The functional $\Phi$ can be seen as a computable set of tuples $\tuple{\sigma,\tau}$
indicating the fact that any real $Y$ with prefix $\sigma$ is $\Phi$-mapped to an extension of $\tau$.
Moreover we can assume that if $\tuple{\sigma,\tau}\in \Phi$ then for each $\rho\subseteq\tau$ there exists
$\eta\subseteq\sigma$ such that $\tuple{\eta,\rho}\in \Phi$. 
Using the fact that $f$ is the oracle use in $\Phi$,
we may also assume that if 
$\tuple{\sigma,\tau}\in \Phi$ then $|\sigma|=f(\tau)$.

For each $i$ consider the set $D_i$ of strings $\tau$ of length $n_i$ such that $\tuple{\sigma,\tau}\in \Phi$
for some $\sigma$. Then $(D_i)$ is uniformly computable, since $\Phi$ is a total Turing functional.
By the choice of $(n_i)$ we have $|\sigma|+i\leq |\tau|$ for each such pair
$\tuple{\sigma,\tau}$. For each $i$ let $M_i$ be the set of strings $\sigma$ such that
$\tuple{\sigma,\tau}\in \Phi$ for some $\tau\in D_i$. According to our hypothesis about $\Phi$,
all strings in $M_i$ have the same length $f(n_i)$, so $M_i$ is prefix-free. So
\[
\mu(D_i)=\sum_{\tau\in D_i} 2^{-|\tau|}\leq \sum_{\sigma\in M_i} 2^{-|\sigma|-i}=
2^{-i}\cdot \sum_{\sigma\in M_i} 2^{-|\sigma|}\leq 2^{-i}.
\] 
Since the sets $D_i$ are also finite and uniformly computable,  
$(D_i)$ is a Kurtz test. By our hypothesis we have $\Phi^Y=X$ for some $Y$. Hence $X\in \cap_i D_i$
and this concludes the proof that $X$ is not Kurtz random.}

\subsection{Proof of Theorems \ref{qTEYegz9as} and \ref{EFigRO1nBQ}}
The proof of Theorem \ref{qTEYegz9as}
 is a direct adaptation of the classic argument that shows the equivalence of the
expressions of \ml randomness in terms of tests, martingales and \pf complexity.
We need to show that
the following are equivalent for a real $X$:
\begin{enumerate}[\hspace{0.5cm}(a)]
\item $X\in \cap_i V_i$ for a computably granular \ml test $(V_i)$;
\item $M(X\restr_{g(n)})\geq n$ for all $n$, a \ce (super)martingale $M$ and a computable function $g$;
\item $K(X\restr_{g(i)})\leq g(i)-i$ for a computable function $g$.
\end{enumerate}

Suppose that (a) holds. Then without loss of generality we may assume that
there exists a computable nondecreasing function $g$ such that for each $i$,
all strings in $V_i$ have length $g(i)$. Then we can define 
$M(\sigma)= \sum_{i} \mu_{\sigma}(V_i)$.\footnote{Here $\mu_{\sigma}(V_i)$ denotes
the measure of $V_i$ relative to the set $[\sigma]$ of reals that have $\sigma$ as a prefix:  
$\mu_{\sigma}(V_i)=\mu_{\sigma}(V_i\cap [\sigma])\cdot 2^{|\sigma|}$.}
It can be easily verified that $M$ is a \ce martingale which meets (b) for any $X\in \cap_i V_i$.
Conversely, assume that (b) holds and define $U_i$ to be the set of strings of length $g(i)$ such that 
$M(X\restr_{g(i)})> i$. Also let $V_i=U_{2^{i}}$ for each $i$. Then $\mu(V_i)\leq 2^{-i}$ and
$(V_i)$ is granular with respect to the function $i\mapsto g(2^i)$.
Now assume that (c) holds. Then define $V_i$ be the set of strings $\sigma$ of length $g(i)$ such that
$K(\sigma)\leq |\sigma|-i$. Then $\mu(V_i)\leq 2^{-i}$ for each $i$ and any $X$ such that 
 $K(X\restr_{g(i)})\leq g(i)-i$ is in $V_i$. Hence we can deduce (a). Conversely, suppose that
 $X\in \cap_i V_i$,  $\mu(V_i)\leq 2^{-i}$ and all strings in $V_i$ have length $g(i)$.
Then using the Kraft-Chaitin theorem we can build a \pf machine $M$ such  that
$K(\sigma)\leq |\sigma|-i$ for all $i$ and for all $\sigma\in V_i$. This means that (a) implies (c).

Next, we turn to the proof of Theorem \ref{EFigRO1nBQ}.
We need to show that the following are equivalent:
\begin{enumerate}[\hspace{0.5cm}(a)]
\item  $X$ is granularly random;
\item for every $Y$ such that $X\leq_{wtt} Y$ the $Y$-use in any such computation of $X\restr_n$
is bounded below by $n-c$ for some constant $c$ and all $n$.
\end{enumerate}
This  proof is very similar to the argument we discussed in Section \ref{DFdfWxKDG}.
\paragraph{Proof that $\neg$(b) implies $\neg$(a).}{
Assume that (b) does not hold. Then $\Phi(Y)=X$ for some $Y$ and a Turing functional $\Phi$ with 
computable non-decreasing oracle-use $f$ on all oracles such that
for each $c$ there exists some $n$ such that $f(n)<n-c$. 
Let $(n_i)$ be a computable increasing sequence such that
$f(n_i)<n_i-i$ for each $i$. The functional $\Phi$ can be seen as a \ce set of tuples
$\tuple{\sigma,\tau}$ indicating that for every oracle $Y$ such that $\Phi(Y)$ is total,
if $Y$ is prefixed by $\sigma$ then $\Phi(Y)$ is prefixed by $\tau$. 
We may also assume that for each $\tuple{\sigma,\tau}$ in $\Phi$ we have
$f(|\tau|)=|\sigma|$.
Hence for every $i$, every $\tau$ of length 
$n_i$ and every $\sigma$ such that $\tuple{\sigma,\tau}\in \Phi$ the length of $\sigma$ is at most
$n_i-i$.
Let $V_i$ be the set of all $\tau$ of length $n_i$ such that
$\tuple{\sigma,\tau}\in \Phi$ for some $\sigma$. 
Then by the choice of $n_i$, there can be at most $2^{n_i-i}$ many distinct strings in $V_i$.
So $\mu(V_i)\leq 2^{-i}$ and $(V_i)$ is a computably bounded \ml test. Moreover $X\in\cap_i V_i$ so
$X$ is not computably bounded random.}

\paragraph{Proof that $\neg$(a) implies $\neg$(b).}{Conversely, assume that $X\in\cap_i V_i$
where $(V_i)$ is a \ml test such that for each $i$ the set $V_i$ contains strings of length $v_i$,
such that  $v_0>1$ and $v_j<v_{j+1}$ for all $j$.
Also, without loss of generality we may assume that for each $i$ we have 
$|V_i|\leq 2^{v_i-i}$ and
every string in $V_{i+1}$ has a prefix in $V_i$.
It suffices to construct a Turing functional $\Phi$ with computable oblivious oracle-use, 
such that for each $X\in \cap_i V_i$ there
exists some $Y$ such that $X=\Phi^Y$.

Consider the tree $T$ constructed in Section \ref{DFdfWxKDG} in terms of 
the Kurtz test $(D_i)$, and replace 
$(D_i)$ with $(V_i)$ and $(d_i)$ with $(v_i)$ in this construction. 
The only difference in the properties of $T$ is that
$T$ is now merely a \ce tree and not necessarily computable.
Similarly to the construction of $T$ in  Section \ref{DFdfWxKDG}, 
consider
the intervals $I_s=[v_{q_s}, v_{q_{s+1}})$.
Define $\Phi$ be the Turing functional 
such that for each $X,s$ and $n\in I_s$ the oracle-use
for the first $n$ bits of $X$ is $v_{q_{s}}-s$ and
\[
\Phi^X(n)=
\begin{cases}
T(X)(n), &\textrm{if $X\in \cap_{j\leq q_{s+1}} V_{j}$};\\
\un & \textrm{otherwise.}
\end{cases}
\]
Then we can conclude that each $X\in \cap_i V_i$ is truth-table reducible to $T(X)$
and for each $s$ the oracle-use for the computation  of the first $v_{q_{s}}$
bits of $X$ is $v_{q_{s}}-s$. This concludes the proof of $\neg$(b) given $\neg$(a).}

\section{Computing and compressing \ce sets and \ce reals}
\subsection{Proof of Theorem \ref{TM6wyBIcXu}}\label{8J69q48pEH}
We wish to prove a slightly stronger version of Theorem \ref{TM6wyBIcXu}. 
We show that given a non-decreasing computable function $f$,
the following are equivalent:
\begin{enumerate}[\hspace{0.5cm}(a)]
\item $\sum_n 2^{-f(n)}$ is finite;
\item every \ce set $X$ is computable by $\Omega$ with oracle-use bounded above by $f$;
\item every \ce set $X$ is computable by some set $Y$ with oracle-use bounded above by $f$.
\end{enumerate}

The implication (a)$\to$(b)
is from \cite{asybound} and
(b)$\to$(c) is trivial. It remains to show  $\neg$(a)$\to\neg$(c). 
So suppose that the sum in (a) is infinite and let $(\Phi_e)$ be an effective enumeration of all 
Turing functionals with 
oblivious oracle-use $f$. It suffices to construct a \ce set $X$ such that 
$X\neq\Phi_e^Y$ for all reals $Y$. Let $(m_e)$ be a computable increasing sequence such that 
\[
\sum_{i\in [m_e,m_{e+1})} 2^{-f(i)}>2
\hspace{1cm}
\textrm{for all $e$}
\]
and let $I_e=[m_e,m_{e+1})$.
The enumeration $(X_s)$ of $X$ is based on the effective enumeration of $(\Phi_e)$.
At stage $s+1$, consider the least $e\leq s$ such that $X_s\restr_{t+1}$ is a prefix of $\Phi_e^{\tau}[s]$ 
for some $t\in I_e-X_s$
and some $\tau$ of length $f(t)$. If such $e$ exists, enumerate the least such $t$ into $X$.

Clearly $X$ is a \ce set. For a contradiction, suppose that $\Phi_e^Y=X$ for some $Y$ and some $e$.
This means that every number in $I_e$ will be enumerated into $X$ at some stage.
Let $s_i, i<|I_e|$ be an increasing enumeration of the stages where these enumerations occurred.
Moreover let $t_i, i<|I_e|$ be the enumerations that occurred in stages $s_i, i<|I_e|$. By standard assumptions about
the functionals $(\Phi_j)$ and the construction, it follows that the sequence $(t_i)$ is increasing.
For each $i<|I_e|$, at stage $s_i$ some number $t_i\in I_e$ was enumerated into $X_{s_i}-X_{s_i-1}$ so
there exists a string $\tau_i$ of length $f(t_i)$ such that  $X_{s_i-1}\restr_{t_i+1}$ is a prefix of 
$\Phi_e^{\tau_i}[s_i-1]$
but $X_{s}\restr_{t_i+1}$ is not a prefix of 
$\Phi_e^{\tau_i}[s_i-1]$
for any $s\geq s_i$. Since $(t_i)$ is increasing and $f$ is non-decreasing, 
it follows that the sequence $(|\tau_i|)$ is non-decreasing and 
$\tau_i,  i<|I_e|$ is a \pf set of strings. Hence by Kraft's inequality we have
\[
\sum_{ i<|I_e|} 2^{-|\tau_i|}\leq 1.
\]
On the other hand by the choice of $I_e$ we have
\[
\sum_{i<|I_e|} 2^{-|\tau_i|}= \sum_{j\in I_e} 2^{-f(j)}>2
\]
which is a contradiction.
Hence we may conclude that there does not exist any real $Y$ such that $\Phi_e^Y=X$ which completes the proof
of $\neg$(a)$\to\neg$(c).

\subsection{Proof of Theorem \ref{Etk37fTeOu}}
Recall the statement of Theorem \ref{Etk37fTeOu}:
\[
\parbox{14cm}{Every \lce real $X$ can be computed from $\Omega$ with oracle-use
$g(n)=\min_{i\geq n} K(X\restr_n)$.}
\]
Let $g,\Omega, X$ be as above and let
$(X_s), (\Omega_s)$ be computable nondecreasing 
dyadic rational approximations that converge to $X,\Omega$ respectively.
We construct a Solovay $J$ test as follows. 

At stage $s+1$, consider the set $I_{s}$ of $n\leq s$ such that
one of the following holds:
\begin{enumerate}[\hspace{0.5cm}(a)]
\item $X_{s+1}\restr_n= X_{s}\restr_n$ and $K(X\restr_n)[s+1]<K(X\restr_n)[s]$;
\item $X_{s+1}\restr_n\neq X_{s}\restr_n$ and $K(X\restr_n)[s+1]<\infty$.
\end{enumerate}
For each $n\in I_s$ enumerate $\Omega\restr_{K(X\restr_n)}[s+1]$ into $J$.

First we verify that $J$ is a Solovay test.
According to the construction, every string $\sigma$ enumerated in $J$ at stage $s+1$ corresponds
to clause (a) or clause (b). In either case, any such $\sigma$ has the length of a corresponding 
currently shortest 
description of $X_{s+1}\restr_n$ for some $n$. Hence any string $\sigma\in J$ is uniquely associated
with a string of the same length, in the domain of the underlying universal machine $U$. Since the weight of
$U$ is bounded by 1, the same will be true of the weight of $J$. Since $J$ is also a \ce set, it is a Solovay test.


It remains to show how to compute $X\restr_n$ from $\Omega$, by a Turing reduction that uses
only the first $g(n)$ bits of of $\Omega$. 
Since $\Omega$ is \ml random there must be a stage $t_0$ in the construction after which no initial segment of 
$\Omega$ is ever enumerated in $J$ after $t_0$. From now on we work in the stages $s>t_0$. 
Note that $g$ is a $\Delta^0_2$ function so it has a
computable approximation $(g_s)$.

The first thing to note is that, by construction, 
\begin{equation}\label{KS9FkftGfO}
\parbox{13cm}{the
settling time of the first $g(n)$ bits of  $\Omega_s$ is larger than the settling time of $g_s(n)$}
\end{equation}
assuming that
these settling times are larger than $t_0$, which is true for all but finitely many $n$.
This is because every time the approximation to $g(n)$ changes at some stage $s+1$, the initial segment
of $\Omega_s$ of length $g_{s}(n)$ is enumerated into $J$, and the fact that after stage $t_0$ no correct
approximation to $\Omega$ is ever enumerated into $J$.

Let $m_t$ be the settling time of $\Omega\restr_t$. In order to compute $X\restr_n$
we look for the first $t$ such that $m_t\geq g_{m_t}(n)$ and $m_t>t_0$. Then we decide that
$X_{m_t}\restr_n$ is the correct approximation to $X\restr_n$.

First note that by \eqref{KS9FkftGfO} we have $t\leq g(n)$, so this computation requires
at most the first $g(n)$ bits of $\Omega$. Second,  this computation is correct since
if $X_{m_t}\restr_n$ was not the correct approximation to $X\restr_n$, the construction would
enumerate $\Omega_{g_{m_t}(n)}\restr_{n}$ into $J$, which means that either
$m_t\geq g_{m_t}(n)$ is not true or $m_t$ is not
the settling time of $\Omega\restr_t$ (either of which is clearly a contradiction).

\subsection{Proof of Theorem \ref{qVH5QUn2qI}}\label{1kYywGW7s}
Given a
computable non-decreasing function
$f$ we need to show that the following are equivalent:\footnote{The assumption that
$f$ is non-decreasing is stronger than the assumption of the original statement of Theorem \ref{qVH5QUn2qI}
that $h(n)=f(n)+\log n$ is nondecreasing. However the reader can verify that the argument we give 
is insensitive to this discrepancy.}
\begin{enumerate}[\hspace{0.5cm}(a)]
\item the \pf complexity of every \ce set is bounded by $f+\bigo{1}$;
\item $\sum_n n\cdot 2^{-f(n)}$ is finite.
\end{enumerate}

Note that (b) implies that $f(n)-\log n$ is an upper bound for $K(n)$. On the other hand,
in order to describe the first $n$ bits of a \ce set we merely need the shortest \pf description of $n$
concatenated by a code of the number of elements in $A\restr_n$.
Since the latter has length at most $\log n$, we see that (b) implies (a) (given the first $K(n)$ bits
we can recover $n$, then calculated $\log n$ and read the next $\log n$, thus executing a self-delimiting code).
%

Next, we show that $\neg$(b) implies $\neg$(a).
Assuming that the sum in (b) is infinite we wish to construct a \ce set $A$ 
whose initial segment \pf complexity is not $f+\bigo{1}$. In the construction of $A$
we will use the following fact.
\begin{lem}\label{eefkq2MrHi}
If $\sum_n n\cdot 2^{-f(n)}$ is infinite then for every $c\in\Nat$ the sum
$\sum_{n>c} (n-c)\cdot 2^{-f(n)}$ is infinite.
\end{lem}
\begin{proof}
For a contradiction suppose that this does not hold for some $c$. Then 
\[
\sum_n n\cdot 2^{-f(n)} - 
\sum_{n>c} (n-c)\cdot 2^{-f(n)}=
\sum_{n\leq c} n\cdot 2^{-f(n)}+
\sum_{n>c} c\cdot 2^{-f(n)}
\] 
is infinite, which means that
\[
\sum_{n>c} 2^{-f(n)}=\infty
\hspace{1cm}
\textrm{and}\hspace{1cm}
\sum_{n>c} (n-c)\cdot 2^{-f(n)}=\infty
\]
which contradicts our assumption.
\end{proof}
We need to construct a \ce set $A$ which satisfies the following requirements:
\[
\mathcal{R}_e:\ \ \exists n\ \ \ K(A\restr_n)>f(n)+e.
\]

By Lemma \ref{eefkq2MrHi} there exists a computable increasing sequence $(n_i)$ such that
\begin{equation}\label{rWfbYOZNJz}
\sum_{s\in (n_e, n_{e+1}]} (s-n_e)\cdot 2^{-f(s)}>2^e
\hspace{1cm}
\textrm{for all $e$.}
\end{equation}

Let us say that $\mathcal{R}_e$ {\em requires attention} at stage $s+1$ if
$K(A\restr_n)[s+1]\leq f(n)+e$ for all $n\leq n_{e+1}$ and $[n_e, n_{e+1})-A_s\neq\emptyset$.
Note that the numbers we are going to enumerate into $A$ for the satisfaction of $\mathcal{R}_e$
are from the interval $[n_e, n_{e+1})$ while in \eqref{rWfbYOZNJz} we use interval 
$(n_e, n_{e+1}]$. This is because for each $e$, a change in the approximation to $A\restr_{n_{e+1}}$ can only
be achieved with the enumeration into $A$ of one of the numbers in $[0, n_{e+1})$.

We can now define a computable enumeration $(A_s)$ of $A$. At stage $s+1$
consider the least $e\leq s$ such that $\mathcal{R}_e$ requires attention. 
If such $e$ exists, enumerate the  least number of $[n_e, n_{e+1})-A_s$ into $A$.
This completes the construction of $A$.

We verify that $A$ meets requirement $\mathcal{R}_e$ for all $e$. 
Since for each $e$ the interval $[n_e, n_{e+1})$ contains finitely many numbers, by definition 
each requirement can only require attention at finitely many stages.
This is because when $e$ is the least number such that  $\mathcal{R}_e$
requires attention at stage $s+1$, then a number from $[n_e, n_{e+1})-A_s$ enters $A$.
If  $\mathcal{R}_e$ was not satisfied for some $e$, then it follows that
$[n_e, n_{e+1})\subseteq A$. But at each stage $s+1$ where a number 
$n\in [n_e, n_{e+1})$ enters $A$ we have
$K(A\restr_{t})[s]\leq f(t)+e$ for all $t\in (n_e, n_{e+1}]$.
So if such an enumeration of 
$n\in [n_e, n_{e+1})$ occurs at some stage stage $s+1$, 
we can count an additional description 
of $A_s\restr_t$ of length at most $f(t)+e$ for each $t\in (n, n_{e+1}]$ in the domain of
in the underlying universal \pf machine. Indeed, this is because if $n\in [n_e, n_{e+1})$ was enumerated into $A$ at
some stage $s+1$,
the previous such enumeration was of the number $n+1$ at some stage $t+1<s+1$ and
for each $m\in (n+2, n_{e+1}]$ we have $A_s\restr_m\neq A_t\restr_m$. Of course if $n\in [n_e, n_{e+1})$
was the first number in this interval to be enumerated into $A$, then $n=n_{e+1}-1$ and we count
the description $A_s\restr_{n_{e+1}-1}$ of length at most $f(n_{e+1})+e$ for the first time. 
This means that by the stage where
all of the numbers in $[n_e, n_{e+1})$ are enumerated into $A$,
we can count descriptions in the domain of the universal machine of weight at least 
\[
\sum_{s\in (n_e, n_{e+1}]} (s-n_e)\cdot 2^{-f(s)-e}.
\]
This contradicts \eqref{rWfbYOZNJz} and the fact that the weight of the domain of a \pf machine is bounded
above by 1.
Hence we can conclude that the \ce set $A$ that we enumerated satisfies
$\mathcal{R}_e$  for each $e$.

\subsection{Proof of Theorem \ref{cJbypbgUJ}}\label{miPvnCjZSZ}
We need to show that:
\begin{enumerate}[\hspace{0.5cm}(a)]
\item if $\mathbf{a}$ is array non-computable then it contains a \ce set whose \pf complexity is not bounded
above by $2\log n+\log\log n$;
\item if $\mathbf{a}$ is array computable then the \pf complexity of all of its \ce members is bounded above by
$\log n+2\log\log n+\bigo{1}$.
\end{enumerate}

For (a), note that $\sum_i 2^{-\log {i}-\log\log i}=\infty$, so we only need to
adapt the construction for the proof of Theorem \ref{qVH5QUn2qI} in Section \ref{1kYywGW7s},
inside any \ce array non-computable \ce degree. 
In the following we fix $f(n)=\log n$  in  Theorem \ref{qVH5QUn2qI}
and the construction of Section \ref{1kYywGW7s}.
Note that it suffices to only satisfy infinitely many requirements
$\mathcal{R}_e$ from Section \ref{1kYywGW7s}.
Recall from \cite{DJS2} that if $(I_n)$ is a computable sequence of intervals such that 
$I_n<I_m$ whenever $n<m$ then 
\begin{equation}\label{A6SpbUfMn1}
\parbox{12cm}{every \ce array non-computable degree
$\mathbf{a}$ contains a \ce set $A$ such that for each \ce set $W$ we have 
$A\cap I_n=W\cap I_n$ for infinitely many $n$.}
\end{equation}
Without loss of generality we may assume that $n_{e+1}-n_e>n_e-n_{e-1}$ for all $e>0$
in the construction of Section \ref{1kYywGW7s}. Then we can set $I_e=[n_e, n_{e+1})$
and apply \eqref{A6SpbUfMn1} in order to fix a \ce set $A$ in $\mathbf{a}$ with the stated property.
Next, we construct a \ce set $W$ in stages as follows, mimicking the construction
of Section \ref{1kYywGW7s}.
At stage $s+1$, let $e\leq s$ be the least number such that 
$K(A\restr_n)[s+1]\leq f(n)+e$ for all $n\leq n_{e+1}$ 
and $I_e\not\subseteq A_s$, and if such $e$ exists 
put the smallest element of $I_e-A_s$ into $W$.

By  \eqref{A6SpbUfMn1}  we have $A\cap I_e=W\cap I_e$ for infinitely many $n$.
Given $e$ such that $A\cap I_e=W\cap I_e$, the assumption that $\mathcal{R}_e$ is not satisfied
leads to a contradiction by the same argument that we used in
Section \ref{1kYywGW7s}. Hence we can conclude that there are infinitely many $e$ such that
$A$ satisfies $\mathcal{R}_e$, which means that the \pf complexity of $A$ is not bounded above by
$\log n + f(n)+\bigo{1}$ which is $2\log n+\bigo{1}$ by our choice of $f$.

For (b), following Kummer's argument from \cite{DBLP:journals/siamcomp/Kummer96} and replacing
plain for \pf complexity, we get that
if $A$ is a \ce set of array computable degree,
$f(n)$ is a computable upper bound of $K(n)$ and
$g$ is a computable unbounded non-decreasing function,
then $K(A\restr_n)$ is bounded above by $f(n)+g(n)+\bigo{1}$.
Note that $\sum_i 2^{-\log {i}-1.5\cdot\log\log i}<\infty$ so if we take $g(n)=0.5\cdot \log\log n$,
the above fact shows that the \pf complexity of $A$ is bounded above by $\log n + 2\log\log n$.

\end{document}